\documentclass[a4paper,10pt]{article}
\usepackage[utf8x]{inputenc}
\usepackage{amssymb}
\usepackage{amsmath}
\usepackage{amsthm}
\usepackage{mathrsfs}
\usepackage{amsfonts}
\usepackage{amsmath}
\usepackage{graphicx}
\usepackage{color}

\newtheorem{theorem}{Theorem}[section]
\newtheorem{proposition}{Proposition}[section]
\newtheorem{remark}{Remark}[section]

\newcommand{\dd}{\mathrm{d}}
\newcommand{\R}{{\mathord{\mathbb R}}}
\newcommand{\C}{{\mathord{\mathbb C}}}

\newcommand{\ee}{\mathrm{e}}

\begin{document}

\title{ Qualitative analysis of magnetic waveguides for two-dimensional Dirac fermions}
\author{{Marie Fialov\'a$^\ddagger$, V\'i{}t Jakubsk\'y$^\dagger$, Mat\v ej Tu\v sek$^\ddagger$}\\
{\small 
\textit{$^\dagger$Department of Theoretical Physics,
Nuclear Physics Institute,
25068 \v Re\v z, Czech Republic}}\\
{\small \textit{$^\ddagger$Department of Mathematics, Czech Technical University in Prague, Trojanova 13, 120 00 Prague, Czech Republic} }\\
\sl{\small{E-mail: fialoma8@fjfi.cvut.cz, jakub@ujf.cas.cz,  tusekmat@fjfi.cvut.cz
%\texttt{at} 
} }}
\date{26.1. 2018}
\maketitle

\begin{abstract}
We focus on the confinement of two-dimensional Dirac fermions within the waveguides created by realistic magnetic fields. Understanding of their band structure is of our main concern. We provide easily applicable criteria, mostly depending only on the asymptotic behavior of the magnetic field, that can guarantee existence or absence of the energy bands and provide valuable insight into the systems where analytical solution is impossible. The general results are employed in specific systems where the waveguide is created by the magnetic field of a set of electric wires or magnetized strips. 
\end{abstract}

\section{Introduction}
A great variety of condensed matter systems can host two-dimensional Dirac fermions. The list of the so called Dirac materials \cite{Wehling} contains not only graphene, but also other systems where relativistic quasi-particles were either predicted or even confirmed experimentally. Let us mention silicene, germanene, dichalcogenides \cite{SiGeTheor1}-\cite{dichalcogenides}, or artificial crystal lattices synthesized in the lab with the use of optical traps  \cite{opt1}-\cite{opt3} or molecular manipulations \cite{molecular},  \cite{agr1}. Due to their remarkable characteristics, Dirac materials (and graphene in particular) are hoped to be the building blocks of the post-silicon electronics. This sets a great motivation to understand and control their electronic properties.

Massless Dirac fermions can propagate through electrostatic barriers without being back-scattered. This phenomenon, Klein tunneling, challenges effective control of the quantum transport with the use of the electric field. Therefore, an alternative approach is needed. Cutting the sample into the required form or chemical adsorptions of other atoms (e.g. oxygen, fluoride) is demanding on precision and lacks flexibility. In this respect, control of the Dirac fermions by magnetic or magneto-electric barriers represents a more feasible option \cite{Rozhkov}.
%it can be very demanding on precision and it lacks flexibility; any change of the circuit would be very problematic. 

Dynamics of two-dimensional (non-relativistic) electron gas (2DEG) in an inhomogeneous magnetic field was studied already in \cite{Muller}, \cite{Reijniers1}, \cite{Matulis}, three-dimensional Dirac fermions in \cite{Malkova}. It was found that when otherwise homogeneous magnetic field changes its sign along a straight line, the electrons can propagate along the line with a non-vanishing group velocity. This is in coherence with the classical picture, where electrons propagate along the trajectories that are bending back and forth along the interface. These states are called snake states in the literature. 

Confinement of Dirac fermions in graphene by magnetic fields was proposed in \cite{DeMartino1}, \cite{DeMartino2}. Propagation both across and along the magnetic barriers was studied in great number of papers so that we have no hope to provide complete list of the related references. The wave vector filtering  known for 2DEG \cite{Matulis}, where the electrons bouncing on the magnetic barrier are totally reflected for a broad interval of angles, was analyzed for Dirac fermions e.g. in \cite{Myoung}, \cite{luxuye}. Existence of snake states in graphene was addressed in \cite{Oroszlany}, \cite{Ghosh}. Effect of many-body interactions was discussed in \cite{Cohnitz1}, spin snake states were discussed in \cite{PeetersSpin}.
The existence of snake states in graphene was also studied both theoretically and experimentally in the systems  where both magnetic and electric fields  change abruptly along a straight line \cite{Bliokh}, \cite{Cohnitz2}, \cite{LiuTiwari}, \cite{Rickhaus}, \cite{Taychat}, see also the review article \cite{Rozhkov}.

The systems studied in the literature usually possess translational invariance in one direction. It implies conservation of the (longitudinal) momentum $k$ along the barrier.  The energy spectrum of the two-dimensional system can be reconstructed  from the spectra of effectively one-dimensional subsystems with fixed $k$. 
In many of the mentioned works, the external fields are represented by rectangular barriers or step functions \cite{Ghosh}, \cite{Matulis}, \cite{Malkova}, \cite{LiuTiwari}, \cite{DeMartino1}, \cite{Reijniers1}, \cite{Cohnitz1}, \cite{Cohnitz2}. As the spectra of their one-dimensional subsystems are purely discrete, the spectrum of the full two-dimensional Hamiltonian consists of energy bands $E_n(k)$ with integer $n$. The energy bands can coincide with Landau levels for large $|k|$ (the case when the magnetic field is asymptotically constant and of the same sign at infinities) or there can also emerge dispersive states localized at the barrier and moving along it with a non-vanishing group velocity $v_n=\partial_{k}E_n(k)$, \cite{Reijniers1}, \cite{Malkova}. The solvable models with step-like vector potential were discussed e.g. in \cite{Myoung}, \cite{Bliokh}. The effectively one-dimensional subsystem has non-vanishing essential spectrum and possibly also the discrete eigenvalues that correspond to the 
states confined at the barrier. The essential spectra build up the typical wedges of allowed energies in the spectrum  of the full Hamiltonian, while the discrete energy levels of the subsystem form the energy bands $E_n(k)$ \cite{Bliokh}.

The energy bands $E_n(k)$ can be associated with the existence of wave packets  that do not disperse in transverse direction \cite{dispersionless}. 
When the group velocity $v_n(k)$ acquires both positive and negative values, the barrier represents a waveguide for \textit{bidirectional} quantum transport of the wave packets. When $v_n(k)$ is strictly positive or negative, then a \textit{unidirectional} quantum transport of the dispersionless wave packets along the barrier takes place \cite{Bliokh}. 

The analysis of solvable models was not restricted to the potentials in the form of step-like functions. The techniques of supersymmetric quantum mechanics were used to study the systems with smooth vector potential \cite{Fernandez}, \cite{Negro}, electrostatic potential \cite{Roy} or their combination \cite{Jak1}. (Quasi-) exactly solvable systems associated with the solutions of Heun equation were discussed in \cite{Portnoi}.

In the current article, we focus on the qualitative analysis of the quantum transport in the magnetic waveguides. We aim at the systems where explicit solutions of the stationary equation do not exist analytically. Our goal is to provide as detailed as possible information on the spectrum of the considered systems. In particular, we will be interested in the existence and properties of the energy bands that are the hallmark of dispersionless wave packets in the systems with translational invariance. Our analysis is based on the variational principle and asymptotic behavior of the vector potential. We will analyze in detail realistic settings where magnetic field is generated either by a set of electric wires or magnetized strips with different direction of magnetization, posed in the proximity of the two-dimensional Dirac fermions. 

The work is organized as follows. In the next section, we rigorously define the considered model and explore its spectral properties in dependence on the asymptotic behavior of the vector potential. We review existing relevant results and introduce the new tools for spectral analysis in the form of statements together with their proofs. 
In the next section, we focus on the bundles of electric wires with zero net-current. Our motivation is two-fold there. First, the results should contribute to our understanding of the behavior of Dirac fermions in the proximity of electric circuits. Second, we show that the bundles of electric wires can be used as tunable magnetic waveguides, where the characteristics of the waveguide can be altered just by changing the electric current in the wires.
Section \ref{sec:strips} is devoted to the analysis of magnetized strips. We consider the cases of parallel and perpendicular magnetizations, and also discuss the situation where alignment of the vector of magnetization is not precise. The last section is left for discussion and outlook.

\section{The model}
We consider a planar system of Dirac fermions in presence of magnetic field $\mathbf{B}$ that is invariant with respect to translations along some axis. For definiteness, let it be the $y$-axis. Therefore, $\mathbf{B}$ may be viewed as a function of $x$-variable only. Moreover, since the model is two-dimensional, $\mathbf{B}$ has non-zero only its transverse $z$-component, which will be denoted by $B$. For our analysis, it is convenient to choose the following gauge
\begin{equation*} 
 \mathbf{A}_{\text{phys}}(x)=(0,A_{\text{phys},y}(x),0),\qquad A_{\text{phys},y}(x):=\int_0^x B(s)\dd s+const.
\end{equation*}
We fix units so that $\hbar=m_e=c=1$. Then the elementary charge is just $e=\sqrt{\alpha}$, where $\alpha$ is the fine structure constant which measures the strength of the coupling. 

In the vicinity of the Dirac points, the system under consideration is described by the following Hamiltonian
\begin{equation}\label{H_A}
 \hat H_{\mathbf{A}}=\sigma_1(-i\partial_x)+\sigma_2(-i\partial_y+A_y(x))+\sigma_3 M,
\end{equation}
where $\sigma_i$ are the Pauli matrices, $M$ is a non-negative constant and
\begin{equation} \label{eq:gauge}
 A_{y}(x):=\sqrt{\alpha}\int_0^x B(s)\dd s+const.
\end{equation}
Remark that, for a concise presentation, we absorbed $\sqrt{\alpha}$ into the definition of $A_y$.  The domain of operator (\ref{H_A}) and its self-adjointness in particular will be discussed at the beginning of Section \ref{sec:general}.

The translational symmetry of the system allows us to understand the system by analyzing one-dimensional subsystems with fixed longitudinal momentum $k_y$. Indeed, using the partial Fourier transform in $y$-variable,
$$(\mathscr{F}_{y\to k_y} \psi)(x,k_y)=\frac{1}{\sqrt{2\pi}}\int_\R\mathrm{e}^{-ik_y y}\psi(x,y)~\dd y,$$
we infer that $H_{\mathbf{A}}$ is unitarily equivalent to
\begin{equation*}
 H_{\mathbf{A}}:=\sigma_1(-i\partial_x)+\sigma_2(k_y+A_y(x))+\sigma_3 M,
\end{equation*}
which decomposes into the direct integral,
\begin{equation*}
 H_{\mathbf{A}}=\int^\oplus_\R H_{\mathbf{A}}[k_y],\quad H_{\mathbf{A}}[k_y]:=\sigma_1(-i\partial_x)+\sigma_2(k_y+A_y(x))+\sigma_3 M,
\end{equation*}
where $H_{\mathbf{A}}[k_y]$ acts in $x$-variable only (with $k_y$ fixed).
More concretely, we have
\begin{equation}\label{Hky}
 H_{\mathbf{A}}[k_y]=\begin{pmatrix}
             M&-i(\partial_x+k_y+A_y)\\
             -i(\partial_x-(k_y+A_y))&-M
          \end{pmatrix},
\end{equation}
for the so-called fiber of the total operator $H_{\mathbf{A}}$. Once we find the spectrum of $H_{\mathbf{A}}[k_y]$ for each $k_y\in\R$, we can reconstruct the spectrum of $H_{\mathbf{A}}$, cf. \cite[Theorem XIII.85]{RSIV}. In particular, $\lambda$ is an eigenvalue of $H_{\mathbf{A}}$ if and only  if the set of all $k_y$ such that $\lambda$ is an eigenvalue of $H_{\mathbf{A}}[k_y]$ is not of zero measure.

To find the spectrum of $H_{\mathbf{A}}[k_y]$, we will employ the fact that  the square of $H_\mathbf{A}[k_y]$ reduces to  a direct sum of two Schr\"{o}dinger operators,
\begin{equation*}
 H_{\mathbf{A}}[k_y]^2=\begin{pmatrix}
	       -\partial_x^2+A_y'+(k_y+A_y)^2+M^2 & 0\\
	       0 & -\partial_x^2-A_y'+(k_y+A_y)^2+M^2
            \end{pmatrix}=:
            \begin{pmatrix}
             h_+[k_y] & 0\\ 0 & h_-[k_y]
            \end{pmatrix}.
\end{equation*}
Then one can construct eigenstates of $H_{\mathbf{A}}[k_y]$ in terms of those of $h_+[k_y]$ and 
$h_-[k_y]$, see e.g. \cite{VJDK}. For the reader's convenience, we will also briefly discuss the exact relation between the eigenvalues and associated eigenstates of $h_\pm[k_y]$ and those of $H_{\mathbf{A}}[k_y]$ in Section \ref{sec:en_bands}.

In this paper, we are primarily interested in the presence and location of discrete eigenvalues, i.e., the isolated eigenvalues of finite multiplicity, in the spectrum of $H_{\mathbf{A}}[k_y]$ in the case when the magnetic field vanishes asymptotically, i.e., 
\begin{equation} \label{eq:B_vanishes}
 \sqrt{\alpha}\lim_{x\rightarrow\pm\infty}B(x)=\lim_{x\rightarrow\pm\infty}A_y'(x)=0.
\end{equation}
It is worth noting that both  $h_+[k_y]$ and $h_-[k_y]$ act like
\begin{equation*}
 h_{\mathrm{Iw}}:= -\partial_x^2+(k_y+A_y(x))^2 + V(x),
\end{equation*}
i.e., like the fiber of a two-dimensional spin-less Schr\"odinger operator with magnetic and electric fields that are invariant with respect to translations in $y$-direction. The case with $V\equiv const.$ has been studied intensively but only when the limits in \eqref{eq:B_vanishes} are non-zero.  In that case, $h_{\mathrm{Iw}}$ has pure point spectrum, i.e., the spectrum consists of isolated eigenvalues of finite multiplicity (in fact, all of them are simple). See \cite{Iw_85} for the original study of A. Iwatsuka and \cite{Tu_16} for an extension to non-constant $V$'s and for an overview of the known results. However, in the present case \eqref{eq:B_vanishes}, the spectrum possesses also a continuous part, cf. Theorem \ref{theo:gen}.

\section{General observations} \label{sec:general}
Spectral properties of any operator are to a large extend determined by its domain of definition. Hence, its specification should precede any rigorous spectral analysis. 
Since $A_y$ is continuous and vanishes at $\pm\infty$, $(H_{\mathbf{A}}[k_y]-H_{\mathbf{0}}[0])$ is a bounded perturbation of $H_{\mathbf{0}}[0]$. Therefore, by the Kato-Rellich theorem \cite[Theorem X.12]{RSII}, $H_{\mathbf{A}}[k_y]$ is selfadjoint on any domain of selfadjointness of $H_{\mathbf{0}}[0]$, which we choose naturally as the first Sobolev space $W^{1,2}(\R;\C^2)$--roughly speaking the  space of $\C^2$-valued square integrable functions with square integrable derivatives. See \cite{Thaller} for a proof of selfadjointness of the latter operator.  Consequently, $H_{\mathbf{A}}$ defined as a direct integral of $H_{\mathbf{A}}[k_y]$ is also selfadjoint \cite[Theorem XIII.85]{RSIV} and, due to unitarity, the same holds true for $\hat H_{\mathbf{A}}$ defined on the partial Fourier preimage of the domain of $H_{\mathbf{A}}$. Remark that this preimage is just $W^{1,2}(\R^2;\C^2)$.

Our starting point for the spectral analysis of $H_{\mathbf{A}}[k_y]$ will be the main result of \cite{VJDK} which we reproduce here not in its full generality but under some special assumptions on the vector potential.

\begin{theorem}[V.~J., D. Krej\v{c}i\v{r}\'{i}k \cite{VJDK}]\label{theo:gen}
 Let $A_y\in C^1(\R)$ satisfy  
 \begin{equation} \label{eq:A_vanishes}
 \lim_{x\to\pm\infty}A_y(x)=0
 \end{equation}
 and \eqref{eq:B_vanishes}.  Moreover, let $A_y(A_y+2k_y)$ be either integrable or there exists $x_0$ such that $A_y(A_y+2k_y)\leq 0$ for all $x<x_0$. Then
 \begin{enumerate}
  \item $E\in\sigma(H_{\mathbf{A}}[k_y])\Leftrightarrow -E\in\sigma(H_{\mathbf{A}}[k_y])$ 
  and $\sigma_{\mathrm{ess}}=(-\infty,-\sqrt{k_y^2+M^2}]\cup[\sqrt{k_y^2+M^2},+\infty)$,
  \item if either 
  \begin{equation} \label{eq:bs_cond1}
  \int_\R A_y(x)(A_y(x)+2k_y)\dd x<0
  \end{equation}
  or there exists $a\in\R$ and $\delta\in\{\pm1\}$ such that 
  \begin{equation} \label{eq:bs_cond2}
  \int_{-\infty}^a A_y(x)(A_y(x)+2k_y)\dd x<-\frac{2}{\sqrt{3}}\sqrt{\sup_{x\in(a,+\infty)}\big(A_y(x)(A_y(x)+2k_y)+\delta A'_y(x)\big)}-\delta A_y(a)
  \end{equation}
  then $H_{\mathbf{A}}[k_y]$ has at least one discrete eigenvalue in $(-\sqrt{k_y^2+M^2},\sqrt{k_y^2+M^2})$, 
  \item $\pm M$ is never an eigenvalue of $H_{\mathbf{A}}[k_y]$,
  \item when $k_y=0$ or if $k_y\neq 0$ and one of $A_y(A_y+2k_y)\pm A_y'$ is non-negative then there are no discrete eigenvalues in the spectrum of $H_{\mathbf{A}}[k_y]$.
 \end{enumerate}
\end{theorem}

\begin{remark}
Note that \eqref{eq:A_vanishes} amounts to an appropriate choice of gauge that may or may not exist.
For example, if $B$ is integrable, one can fix the constant in \eqref{eq:gauge} so that $\lim_{x\to+\infty}A_y(x)=0$. Then $\lim_{x\to-\infty}A_y(x)=0$ if and only if
\begin{equation} \label{eq:int_B_vanishes}
\int_\R B(x) \dd x=0,
\end{equation}
which holds true, e.g., for all integrable odd magnetic fields $B$. 
\end{remark}

\begin{remark} \label{rem:gauge}
 The condition \eqref{eq:bs_cond2} can be modified for a neighbourhood of $+\infty$ as follows
 \begin{equation*}
  \int_{a}^{+\infty} A_y(x)(A_y(x)+2k_y)\dd x<-\frac{2}{\sqrt{3}}\sqrt{\sup_{x\in(-\infty,a)}\big(A_y(x)(A_y(x)+2k_y)+\delta A'_y(x)\big)}+\delta A_y(a),
 \end{equation*}
 whenever $A_y(A_y+2k_y)$ is either integrable or there exists $x_0$ such that $A_y(A_y+2k_y)\leq 0$ for all $x>x_0$.
\end{remark}

Let us stress that from now on we will only deal with magnetic fields that satisfy the general assumptions of Theorem \ref{theo:gen}, i.e.,  \eqref{eq:B_vanishes} and \eqref{eq:A_vanishes}.

\subsection{Basic properties of energy bands} \label{sec:en_bands}

Clearly, the square of any eigenvalue of $H_{\mathbf{A}}[k_y]$ is an eigenvalue of $H_{\mathbf{A}}[k_y]^2$. Conversely, the eigenvalues of $H_{\mathbf{A}}[k_y]^2$ are given exactly as the union of the eigenvalues of $h_\pm[k_y]$. Let us investigate their point spectra in more detail.
It is convenient to rewrite $h_{\pm}[k_y]$ in factorized form, $h_+[k_y]=pp^\dagger+M^2$ and $h_-[k_y]=p^\dagger p+M^2$ with $p:=-i(\partial_x+k_y+A_y)$. As $h_{\pm}[k_y]\geq M^2$, their eigenvalues cannot be smaller than $M^2$.  Moreover, by direct inspection, one can verify that assuming \eqref{eq:A_vanishes} no non-zero solution $f$ of the eigenvalue problem $h_\pm[k_y]f=M^2f$ is square integrable. Therefore, $M^2$  may not be an eigenvalue of $h_\pm[k_y]$. The point spectra of the two operators are identical. Indeed, if $h_+[k_y]\psi=\lambda\psi$ (with $\lambda>M^2$) then $h_-[k_y](p^\dagger\psi)=\lambda p^\dagger\psi$, and if $h_-[k_y]\varphi=\lambda\varphi$ then $h_+[k_y](p\varphi)=\lambda p\varphi$.  As $h_\pm[k_y]$ is a Sturm-Liouville operator, the 
wronskian of two solutions to the eigenvalue problem is constant.  Since these solutions have to belong to the domain of $h_\pm[k_y]$, which, in particular, contains only functions vanishing at $x=\pm\infty$ together with their first derivatives, this constant must be zero.  Consequently, one solution is a constant multiple of the other. Hence, all eigenvalues of $h_\pm[k_y]$ are simple. 

The spinors $\Psi_\pm:=(\psi,(M\pm\sqrt{\lambda})^{-1}p^\dagger\psi)^T$ and $\tilde\Psi_\pm:=((-M\pm\sqrt{\lambda})^{-1}p\varphi,\varphi)^T$ are eigenstates of $H_{\mathbf{A}}[k_y]$ with eigenvalue $\pm\sqrt{\lambda}$. These eigenvalues are simple, i.e.,  $\Psi_\pm$ and $\tilde\Psi_\pm$  may differ only by a multiplication constant. Indeed, as $\lambda$ is a simple eigenvalue of both  $h_\pm[k_y]$, we can choose  $\varphi=p^\dagger\psi$. Consequently, $\psi=(\lambda-M^2)^{-1}p\varphi$ and $\tilde\Psi_\pm=(M\pm\sqrt{\lambda})\Psi_\pm$. 
We conclude that having the point spectrum of $h_+[k_y]$ (or $h_-[k_y]$) we can completely reconstruct the point spectrum of $H_{\mathbf{A}}[k_y]$. Moreover, the latter  is simple and localized symmetrically in $(-\sqrt{k_y^2+M^2},-M)\cup(M,\sqrt{k_y^2+M^2})$.

Due to the minimax principle, the eigenvalues of $h_+[k_y]$ may only accumulate at the bottom of its essential spectrum, which was derived in \cite{VJDK} to be $\sigma_{ess}(h_+[k_y])=[k_y^2+M^2,+\infty)$. Therefore, we can enumerate them in increasing order as $E_n[k_y],\, n=1,2,\ldots N$, where $N$ may be possibly infinity and in general it varies with $k_y$. Remark that there are no embedded eigenvalues in the essential spectrum because we already know that every eigenvalue is simple. Therefore, every eigenvalue is necessarily discrete. The energy bands of $H_{\mathbf{A}}$ are then given as the functions $k_y\mapsto\pm\sqrt{E_n[k_y]}$ on some open intervals. Above, we deduced that they may not intersect each other (it would contradict non-degeneracy of the eigenvalues) and cannot pass via the gap $\R\times[-M,M]$ or cross the curves $k_y\mapsto\pm\sqrt{k_y^2+M^2}$, which enclose the essential spectra of $H_{\mathbf{A}}[k_y]$. Moreover, since $H_{\mathbf{A}}[k_y]$ depends analytically on $k_y$ (in fact, it 
forms the so-called analytic family of type (A), see \cite{RSIV}), the energy bands of $H_\mathbf{A}$ are 
analytic on their domains of definition.

The  Hellmann-Feynman formula says that
\begin{equation} \label{eq:HF}
\frac{\dd}{\dd k_y}E_n[k_y]=2\int_\R(k_y+A_y(x))|\psi_n[k_y](x)|^2\dd x,
\end{equation}
where $\psi_n[k_y]$ stands for the normalized eigenfunction of $h_+[k_y]$ associated with $E_n[k_y]$. Therefore, for all $k_y$ sufficiently positive,  $E^\pm_n[k_y]$ is strictly increasing, and for all $k_y$ sufficiently negative, $E^\pm_n[k_y]$ is strictly decreasing. This implies that the energy bands of $H_{\mathbf{A}}$ tend either to $+\infty$ or $-\infty$ as $|k_y|\to\infty$. In particular, no energy band  may be constant; not even locally, due to analyticity. Consequently, there are never eigenvalues in the spectrum of the total, two-dimensional, operator $H_{\mathbf{A}}$.

Finally, if, for any $k_y>0$, there are infinitely many eigenvalues of $H_{\mathbf{A}}[k_y]$ then all respective energy bands are defined on the whole positive half-axis, i.e. they emerge only from $k_y=0$ and may not disappear in the essential spectrum ( which is $(-\infty,-\sqrt{k_y^2+M^2}]\cup[\sqrt{k_y^2+M^2},+\infty)$, in view of Theorem \ref{theo:gen}) for some positive value of $k_y$. Indeed, assume, e.g., that at some definite $k_y>0$ a new pair of  energy bands emerges from the endpoints of the essential spectrum, i.e., $\pm\sqrt{k_y^2+M^2}$. If we increase $k_y$ slightly then the energy bands that existed also for smaller values of $k_y$ are sandwiched between the new pair. Since there is infinitely many of them in a bounded interval, they have to accumulate somewhere within its closure. But the only possible accumulation points given by $\pm\sqrt{k_y^2+M^2}$ are separated from the interval, which is a contradiction. The same holds true on the negative semi-axis.

\subsection{On existence and number of energy bands}

In this section, we present several conditions for existence or non-existence of eigenvalues of $H_{\mathbf{A}}[k_y]$, which are chiefly derived from Theorem \ref{theo:gen} but are much easier to apply then those of Theorem \ref{theo:gen} directly. Furthermore, in some cases we will be able to decide whether there is either finite or infinite number of them. Let us start with a criterion for slowly decaying vector potentials.

\begin{proposition} \label{prop:slowly}
 Let $A_y$ fulfill the general assumptions of Theorem \ref{theo:gen}. 
\begin{itemize}
\item If there exists $x_0\in\R$ such that for all $x<x_0$, $A_y(x)\geq 0$ or $A_y(x)\leq 0$, respectively, and $A_y$ is not integrable on $(-\infty,x_0)$ then for any $k_y<0$ or any $k_y>0$, respectively, $H_{\mathbf{A}}[k_y]$ has infinite number of discrete eigenvalues. 
\item If there exists $x_0\in\R$ such that for all $x>x_0$, $A_y(x)\geq 0$ or $A_y(x)\leq 0$, respectively, and $A_y$ is not integrable on $(x_0,+\infty)$
then for any $k_y<0$ or any $k_y>0$, respectively, $H_{\mathbf{A}}[k_y]$ has infinite number of discrete eigenvalues. 
\end{itemize}
\end{proposition}
\begin{proof}
For definiteness, we consider the case $A_y(x)\geq 0$ for all $x<x_0$, and we choose $k_y<0$. Firstly, take $\varepsilon>0$ such that $(\varepsilon+2k_y)<0$. Due to \eqref{eq:A_vanishes}, one can find $x_1\in\R$ such that $x_1<x_0$ and, for all $x<x_1$,
$0\leq A_y(x)<\varepsilon$.  We have
$$\int_{-\infty}^{x_1}A_y(A_y+2k_y)\leq(\varepsilon+2 k_y)\int_{-\infty}^{x_1}A_y=-\infty.$$
Therefore, \eqref{eq:bs_cond2} with $a=x_1$ is clearly fulfilled, and so there is at least one eigenvalue of $H_{\mathbf{A}}[k_y]$ in $(M,\sqrt{k_y^2+M^2})$ (and, due to the spectral symmetry, at least one eigenvalue in $(-\sqrt{k_y^2+M^2},-M)$). Inspecting the proof of Theorem \ref{theo:gen}, which is based on an explicit construction of appropriate test function for $h_\pm$, one can see that it is even possible to construct infinite number of  test functions with   pair-wise disjoint supports. Consequently, there are infinitely many discrete eigenvalues. The other cases are treated similarly.
\end{proof}

If the vector potential is integrable, i.e. $\int_\R |A_y|<+\infty$, the same holds true for its square, because $A_y$ is bounded under our assumptions. Then \eqref{eq:bs_cond1} may be 
reformulated as follows.

\begin{proposition} \label{prop:fastly}
Suppose $A_y$ fulfills the general assumptions of Theorem \ref{theo:gen} and is integrable. 
\begin{itemize}
 \item 
If $\int_\R A_y>0$ then $H_{\mathbf{A}}[k_y]$ has some discrete eigenvalues for all $k_y<-\frac{\int_\R A_y^2}{2\int_\R A_y}$.
\item 
If $\int_\R A_y<0$ then $H_{\mathbf{A}}[k_y]$ has some discrete eigenvalues for all $ k_y>\frac{\int_\R A_y^2}{2\big|\int_\R A_y\big|}$.
\end{itemize}
\end{proposition}

If $A_y$ does not change its sign then we have a more explicit, though generally rougher, estimate for the threshold value of $k_y$. 

\begin{proposition} \label{prop:const_sign}
Let us suppose that  $A_y$ fulfills the general assumptions of Theorem \ref{theo:gen}.
\begin{itemize}
 \item If $A_y\geq 0$ and $A_y\not\equiv 0$ then $H_{\mathbf{A}}[k_y]$ has some discrete eigenvalues for all $k_y<-\frac{1}{2}\max A_y$.
\item If $A_y\leq 0$ and $A_y\not\equiv 0$ then $H_{\mathbf{A}}[k_y]$ has some discrete eigenvalues for all $ k_y>\frac{1}{2}|\min A_y|$.
\end{itemize}
\end{proposition}
\begin{proof}
 The proof follows directly from \eqref{eq:bs_cond1}, where we require either $A_y+2k_y<0$ or $A_y+2k_y>0$, respectively.
\end{proof}

If one does not care about an explicit range of the values of $k_y$ for which  discrete eigenvalues of  $H_{\mathbf{A}}[k_y]$ exist then one arrives to the following conclusion.

\begin{proposition} \label{prop:gen}
 Let $A_y$ fulfill the general assumptions of Theorem \ref{theo:gen}. 
\begin{itemize}
\item If, for some $a\in\R$, $\int_{-\infty}^a A_y<0$ or $\int_{-\infty}^a A_y>0$, respectively, then there exists $K>0$ such that for all $k_y>K$ or $k_y<-K$, respectively, $H_{\mathbf{A}}[k_y]$ has some discrete eigenvalues.
\item If, for some $a\in\R$, $\int_{a}^\infty A_y<0$ or $\int_{a}^\infty A_y>0$, respectively, then there exists $K>0$ such that for all $k_y>K$ or $k_y<-K$, respectively, $H_{\mathbf{A}}[k_y]$ has some discrete eigenvalues.
\end{itemize}
\end{proposition}
\begin{proof} Let us prove the first statement as the proof of the second one is almost identical.
 One can observe that the left-hand-side and the right-hand-side of \eqref{eq:bs_cond2} behave differently as $|k_y|\to+\infty$. Indeed,
 we have
 \begin{equation} \label{eq:as1}
 \int_{-\infty}^a A_y(A_y+2k_y)=2\left(\int_{-\infty}^a A_y\right)k_y+\mathcal{O}(1)
 \end{equation}
 as $k_y\to\pm\infty$. On the other hand, 
 $$0\leq\sup_{x\in(a,+\infty)}\big(A_y(x)(A_y(x)+2k_y)+\delta A'_y(x)\big)\leq \sup_{x\in(a,+\infty)}A_y(x)^2+2\sup_{x\in(a,+\infty)}|A_y(x)||k_y|+\sup_{x\in(a,+\infty)}|A'_y(x)|.$$
 Hence, there is a lower bound for the right-hand-side of \eqref{eq:bs_cond2} with the asymptotic expansion
 \begin{equation} \label{eq:as2}
  -2\sqrt{\frac{2}{3}}~\sqrt{\sup_{x\in(a,+\infty)}|A_y(x)|}\sqrt{|k_y|}+\mathcal{O}(1) 
 \end{equation}
 as $k_y\to\pm\infty$. If $k_y$ is of the opposite sign than $\int_{-\infty}^a A_y$ then, for all $|k_y|$ sufficiently large, \eqref{eq:as1} is surely below \eqref{eq:as2}.   
\end{proof}

Remark  that whenever $B\not\equiv 0$, there exists $a\in\R$ such that $\int_{-\infty}^a A_y\neq 0$.
If we had $\int_{-\infty}^a A_y= 0$ for all $a\in\R$ then by taking the derivative with respect to $a$ of the  integral we see that $A_y\equiv 0$, and so $B=A_y'\equiv 0$. Therefore, there are  always some energy bands in the spectrum of $H_{\mathbf{A}}$.

Beside our criteria let us reproduce here a known result on the Schr\"odinger operators that can be found in \cite{Teschl} and that we applied on $h_\pm[k_y]$. Recall that $k_y^2+M^2$ is the minimum of the essential spectrum of $h_\pm[k_y]$.

\begin{proposition}[\cite{Teschl}, Corollary 9.43] \label{prop:teschl} Suppose $A_y$ obeys \eqref{eq:B_vanishes} and \eqref{eq:A_vanishes}. If
\begin{equation} \label{eq:teschl1}
\limsup_{|x|\to+\infty}\big(x^2 (A_y(A_y+2k_y)\pm A'_y(x))\big)<-\frac{1}{4}
\end{equation}
then $h_\pm[k_y]$ has infinitely many discrete eigenvalues below $k_y^2+M^2$, and
if 
\begin{equation} \label{eq:teschl2}
\liminf_{|x|\to+\infty}\big(x^2 (A_y(A_y+2k_y)\pm A'_y(x))\big)>-\frac{1}{4}
\end{equation}
then $h_\pm[k_y]$ has finitely many (and possibly none) discrete eigenvalues below $k_y^2+M^2$. 
\end{proposition}

Recall that if $\lambda$ is an eigenvalue of $h_\pm[k_y]$ then $\pm\sqrt{\lambda}$ are eigenvalues of $H_{\mathbf{A}}[k_y]$. Thus, we can use  \eqref{eq:teschl1} and \eqref{eq:teschl2} directly as sufficient conditions for the existence of infinite or finite (possibly zero) number of eigenvalues of $H_{\mathbf{A}}[k_y]$. To evaluate these conditions it is sufficient to know the asymptotic expansions of $A_y$ and $A_y'$ at $\pm\infty$. In fact, the same holds true  when using Proposition \ref{prop:slowly} or \ref{prop:gen}. Although in some situations our criteria may overlap with Proposition \ref{prop:teschl}, they are not a special case of the latter (and vice versa). For instance, notice that in Proposition \ref{prop:teschl} we require some type of asymptotic behavior at both endpoints $\pm\infty$, whereas in Proposition \ref{prop:slowly} or \ref{prop:gen} we are interested only in either of the endpoints. Also beware that employing  \eqref{eq:teschl2} we cannot distinguish between none and finitely 
many eigenvalues, whereas Proposition \ref{prop:gen} can guarantee their existence. 

In addition to the criteria on existence of discrete eigenvalues, it is useful to present a sufficient condition on emptiness of the point spectrum of $H_{\mathbf{A}}[k_y]$.

\begin{proposition} \label{prop:no_ev}
Suppose $A_y$ obeys \eqref{eq:B_vanishes} and \eqref{eq:A_vanishes}.
\begin{itemize}
\item If $A_y\geq 0$ then, for all $k_y\geq 0$, there are no eigenvalues in the spectrum of $H_{\mathbf{A}}[k_y]$.
\item If $A_y\leq 0$ then, for all $k_y\leq 0$, there are no eigenvalues in the spectrum of $H_{\mathbf{A}}[k_y]$.
\end{itemize}
\end{proposition}

\begin{proof}
We will prove only the first statement, as the proof of the second one is almost identical.
From Section \ref{sec:en_bands}, we already know that $H_{\mathbf{A}}[k_y]$ has eigenvalues if and only if  $h_+[k_y]$ has some. The eigenvalues $E_n[k_y]$ of the latter operator may only be found in the interval $(M^2,M^2+k_y^2)$ and the corresponding energy bands $E_n:k_y\mapsto E_n[k_y]$ can only emerge from the bottom of essential spectrum. Suppose there is such energy band emerging at $\tilde{k}_y\geq 0$ from the bottom of the essential spectrum $M^2+\tilde k_y^2$. By \eqref{eq:HF}, we have
$$\frac{\dd E_n}{\dd k_y}[\tilde k_y]\geq 2\tilde k_y,$$
because $A_y\geq 0$ and $\psi_n[k_y]$ is normalized. But in the same moment, we get
$$\frac{\dd}{\dd k_y}\min\big(\sigma_{ess}(h_+[k_y])\big)=2k_y.$$
This means that the energy band $E_n$  grows at least as fast as the threshold of the essential spectrum and, hence, it is out of the allowed region, which is a contradiction.
\end{proof}

At the end of the section,  several remarks are in order. 
As we mentioned briefly in the introduction, the spectral bands can be associated with dispersionless wave packets. They are constructed as follows. 
For definiteness, let us consider a positive energy band $\sqrt{E_n}$ that exists for any fixed $k_y$ from an open interval $J\subset\R$. There is an associated normalized bound state $\Psi_n[k_y]$ that satisfies 
\begin{equation}\label{H(x,k)E}
(H_A[k_y]-\sqrt{E_n[k_y]})\Psi_n[k_y]=0,\quad \forall k_y\in J.
\end{equation} 
Now, we can construct a wave function $\Psi_n=\Psi_n(x,y)$ using the bound states $\Psi_n[k_y]=\Psi_n[k_y](x)$, 
\begin{equation}\label{Psi_n}
 \Psi_n(x,y)=\frac{1}{\sqrt{2\pi}}\int_{I_n}e^{ik_y y}\beta_n(k_y)\Psi_n[k_y](x)\dd k_y,\quad I_n\subset J.
\end{equation}
The coefficient function $\beta_n(k)$ determines the actual form of the wave packet.  We suppose $\int_{I_n}|\beta_{n}(k_y)|^2dk_y=1$ which guarantees normalization of $\Psi_n$. The wave function $\Psi_n(x,y)$ is dispersionless in transverse direction, i.e., its density of probability along $x$-axis remains intact during the evolution, 
\begin{multline*}
 \int_\R|\ee^{-it\hat H_{\mathbf{A}}}\Psi_n(x,y)|^2\dd y=\int_{I_n}|\ee^{-itH_{\mathbf{A}}}\Psi_n(x,k_y)|^2\dd k_y=\int_{I_n}|\beta_n(k_y)\ee^{-it\sqrt{E_n[k_y]}}\Psi_n[k_y](x)|^2\dd k_y\\
 =\int_\R|\Psi_n(x,y)|^2\dd y.
\end{multline*} 
The speed of the wave packet can be then approximated by the averaged group velocity 
$$v=\frac{\int_{I_n}\big(\sqrt{E_n[k_y]}\big)'\dd k_y}{|I_n|},$$
cf. \cite{dispersionless}.

Now, let us suppose that $A_y$ does not change sign. When it is everywhere positive, then Proposition \ref{prop:fastly} or  Proposition \ref{prop:const_sign}  guarantee existence of discrete energies for any $k_y<k_-<0$, where the explicit value of $k_-$ is specified in different manner by each of the criteria. Simultaneously, Proposition \ref{prop:no_ev} tells that there are no discrete energies for all $k_y\geq0$.  Hence, the band structure of the system with such $A_y$ is substantially asymmetric with respect to $k_y$. As it is suggested by the formula (\ref{eq:HF}),  the energy bands are decreasing for sufficiently large negative values of $k_y$. The speed of the associated dispersionless wave packets is then negative. Hence, there are wave packets moving in one direction but one cannot construct similar wave packets moving in the other direction. (Notice that the wave packets assembled from the bound states corresponding to the negative energies are composed from holes, and, therefore, contribute to the electric current in the same direction as the wave packets assembled from the bound states with positive energies.) A similar spectral
asymmetry can be found in the systems with negative $A_y$. We conclude that when $A_y$ does not change sign, the system can host \textit{unidirectional dispersionless wave packets} moving along $y$-axis. 

On the other hand, if the vector potential $A_y$ changes sign asymptotically, the Proposition \ref{prop:gen} tells us that there are energy bands for all sufficiently large $|k_y|$. In  that case, we can construct dispersionless wave packets that will propagate in both directions along the wave guide. The system will host \textit{bidirectional dispersionless wave packets}. 

When the system possesses additional symmetries, they can be reflected in its spectrum. Let us take the system with even magnetic field, and, hence, with odd $A_y$. Then $h_\pm[k_y]$ is unitarily equivalent to $h_\pm[-k_y]$ via the parity transform $P:\psi(x)\mapsto\psi(-x)$. Therefore, their spectra coincide. Consequently, the eigenvalues of $H_{\mathbf{A}}[k_y]$ coincide with those of $H_{\mathbf{A}}[-k_y]$, too. Putting this together with the first point of Theorem \ref{theo:gen}, we conclude that the family of energy bands of $H_\mathbf{A}$ is symmetric with respect to reflections over horizontal and vertical axes.

\section{Wires\label{sec:wires}}
Let us have an infinite wire of negligible radius that runs parallelly with the $y$-axis. It can be parametrized as $y\mapsto(x_0,y,z_0)$ where $x_0\in\R$ and $z_0\in\R\setminus\{0\}$. The magnetic vector potential along the sheet of the Dirac material (placed in $(x,y)$-plane) induced by the wire carrying the current $I$ in the negative $y$-direction can be computed as $\mathbf{A}=(0,A_y,0)$ with
\begin{equation}\label{Aw1}
A_y(x)=\beta I\ln((x-x_0)^2+z_0^2),
\end{equation}
where
\begin{equation*}
\beta:=\frac{\mu_0\sqrt{\alpha}}{4\pi}. 
\end{equation*}
Thus we have
\begin{equation*}
B(x)=\frac{1}{\sqrt{\alpha}}A_y'(x)=\frac{\mu_0 I}{2\pi}\frac{(x-x_0)}{(x-x_0)^2+z_0^2}.
\end{equation*}
With this choice of $A_y$, the potential term in $h_\pm[k_y]$ tends to infinity as $x\to\pm\infty$. Hence, the spectrum of $h_{\pm}[k_y]$ and, consequently, that of  $H_{\mathbf{A}}[k_y]$ is purely discrete~\cite[Theorem XIII.67]{RSIV}. So there are always infinitely many eigenvalues for any value of $k_y\in\R$.

We will now consider a configuration of  parallel wires with tunable asymptotics of $A_y$. Each of the wires will be parametrized as $y\mapsto(x_j,y,z_j)$ and will carry the current $I_j$. We will always require that the total electric current is zero, i.e., 
\begin{equation}\label{zeronetI} 
\sum_{j=1}^{N}I_j=0.
\end{equation}
Under this condition, the total vector potential diminishes at $x=\pm\infty$. Violating \eqref{zeronetI} leads to the purely discrete spectrum as in the case of the single wire.

If \eqref{zeronetI} holds then the vector potential associated with the system of the wires obeys
\begin{equation} \label{eq:wire_asy}
A_y(x)=\beta \sum_{j=1}^N  I_j\ln((x-x_j)^2+z_j^2)=\frac{\alpha_1}{x}+\frac{\alpha_2}{x^2}+\frac{\alpha_3}{x^3}+\dots
\end{equation}
as $x\to\pm\infty$. Here the the coefficients $\alpha_k$ are functions of the coordinates $(x_j,z_j)$ and currents $I_j$, $j=1,\dots N$. For instance, we have 
\begin{equation*}
 \alpha_1=-2\beta \sum_{j=1}^N I_j x_j,\quad \alpha_2=\beta\sum_{j=1}^N I_j(z_j^2-x_j^2).
\end{equation*}
Since \eqref{eq:wire_asy} is analytic at $\pm\infty$, we can differentiate it term by term. In this way, we obtain
\begin{equation*}
B(x)=\frac{1}{\sqrt{\alpha}}A_y'(x)=-\frac{1}{\sqrt{\alpha}}\left(\frac{\alpha_1}{x^2}+\frac{2\alpha_2}{x^3}+\frac{3\alpha_3}{x^4}+\ldots\right) 
\end{equation*}
as $x\to\pm\infty$.

Firstly, let $\alpha_1\neq 0$. Then $A_y$ is not integrable near infinities and it is positive near $+\infty$ and negative near $-\infty$ (or vice versa, depending on the sign of $\alpha_1$). Using Proposition \ref{prop:slowly} we immediately infer that, for any $k_y\neq 0$, $H_{\mathbf{A}}[k_y]$ has infinite number of discrete eigenvalues.

Secondly, let $\alpha_1=0$ but $\alpha_2\neq0$. For definiteness, take $\alpha_2<0$. By Proposition \ref{prop:gen}, $H_{\mathbf{A}}[k_y]$ has discrete eigenvalues for all sufficiently large $k_y>0$. If $\alpha_2>0$ then the discrete eigenvalues are guaranteed for all large enough negative $k_y's$. In both cases, the number of eigenvalues may be finite or infinite, depending on the particular values of $\alpha_2$ and $k_y$, see Proposition \ref{prop:teschl}. As the vector potential does not change sign asymptotically, the system is a good candidate for hosting the unidirectional dispersionless wave packets. However, we are able to prove their existence only when $A_y$ is of constant sign everywhere, cf. Proposition \ref{prop:no_ev}. In Section \ref{dipole}, we will discuss a specific example with $A_y$ of such type.

Finally, for $\alpha_1=0$, $\alpha_2=0$, and $\alpha_3\neq 0$, we have
\begin{equation*}
A_y(x)=\frac{\alpha_3}{x^3}+\mathcal{O}\left(\frac{1}{x^4}\right)
\end{equation*}
as $x\to\pm\infty$. Therefore, due to Proposition \ref{prop:teschl}, there may be, in principle, at most finitely many eigenvalues in the spectrum of $H_{\mathbf{A}}[k_y]$. Employing Proposition \ref{prop:gen} we see that there are definitely some eigenvalues of $H_{\mathbf{A}}[k_y]$ for all $k_y$ above some positive threshold and all $k_y$ below some negative threshold.

In Section \ref{sec:tunable}, we will propose a simple configuration of wires which will make it possible to switch between different types of asymptotic behavior of $A_y$ just by setting the currents appropriately.
Notice that in real experiments, it is rather impossible to place the wires at exact positions and tune the currents so that $\alpha_k=0$ for all required $k$'s. Imprecision in positioning of the wires prevents the coefficients $\alpha_k$ from vanishing, the vector potential behaves asymptotically as $\alpha_1/x$, and $H_{\mathbf{A}}[k_y]$ has infinite number of discrete eigenvalues for any $k_y\neq 0$. Despite these facts, the situation is actually not that bad. In the next section we will deal with a particular configuration of wires that may host unidirectional  dispersionless wave packets. Although this will change even under a small perturbation of the configuration, we will demonstrate that the eigenstates that emerge due to the perturbation are well localized very far from the wires.

\subsection{Pair of wires \label{dipole}}

Let us now focus on situation when $N=2$, i.e., we have pair of wires. The condition \eqref{zeronetI} then reads $I_1=-I_2$. Hence, equivalently, one can think about an infinite single loop. 
The vector potential generated by the loop is given by $\mathbf{A}=(0,A_y,0)$ with
\begin{equation} \label{Adip1}
A_y(x,z)=\beta I_1\ln\frac{(x-x_1)^2+(z-z_1)^2}{(x-x_2)^2+(z-z_2)^2}.
\end{equation}
We fix the distance between the wires to be $2d$ for some $d>0$. It is convenient to choose the coordinate frame so that
\begin{equation} \label{eq:loop_pos}
 x_1=d\cos\gamma,\quad x_2=-d\cos\gamma,\quad z_1=-d\sin\gamma,\quad z_2=d\sin\gamma,
\end{equation}
where $\gamma\in[0,2\pi)$. If we identify the sheet of the Dirac material with the plane $z=z_0$, then $\gamma$ measures the angle between the sheet and the plane defined by the loop, see Figure \ref{fig:loop}. Substituting $z=z_0$ and \eqref{eq:loop_pos} into \eqref{Adip1} we obtain the vector potential along the sheet,
\begin{equation} \label{eq:A_rotated}
 A_y(x)=\beta I_1\ln\frac{(x-d\cos\gamma)^2+(z_0+d\sin\gamma)^2}{(x+d\cos\gamma)^2+(z_0-d\sin\gamma)^2}.
\end{equation}

%%%%%%%%%
\begin{figure} 
\centering
\includegraphics[width=7cm]{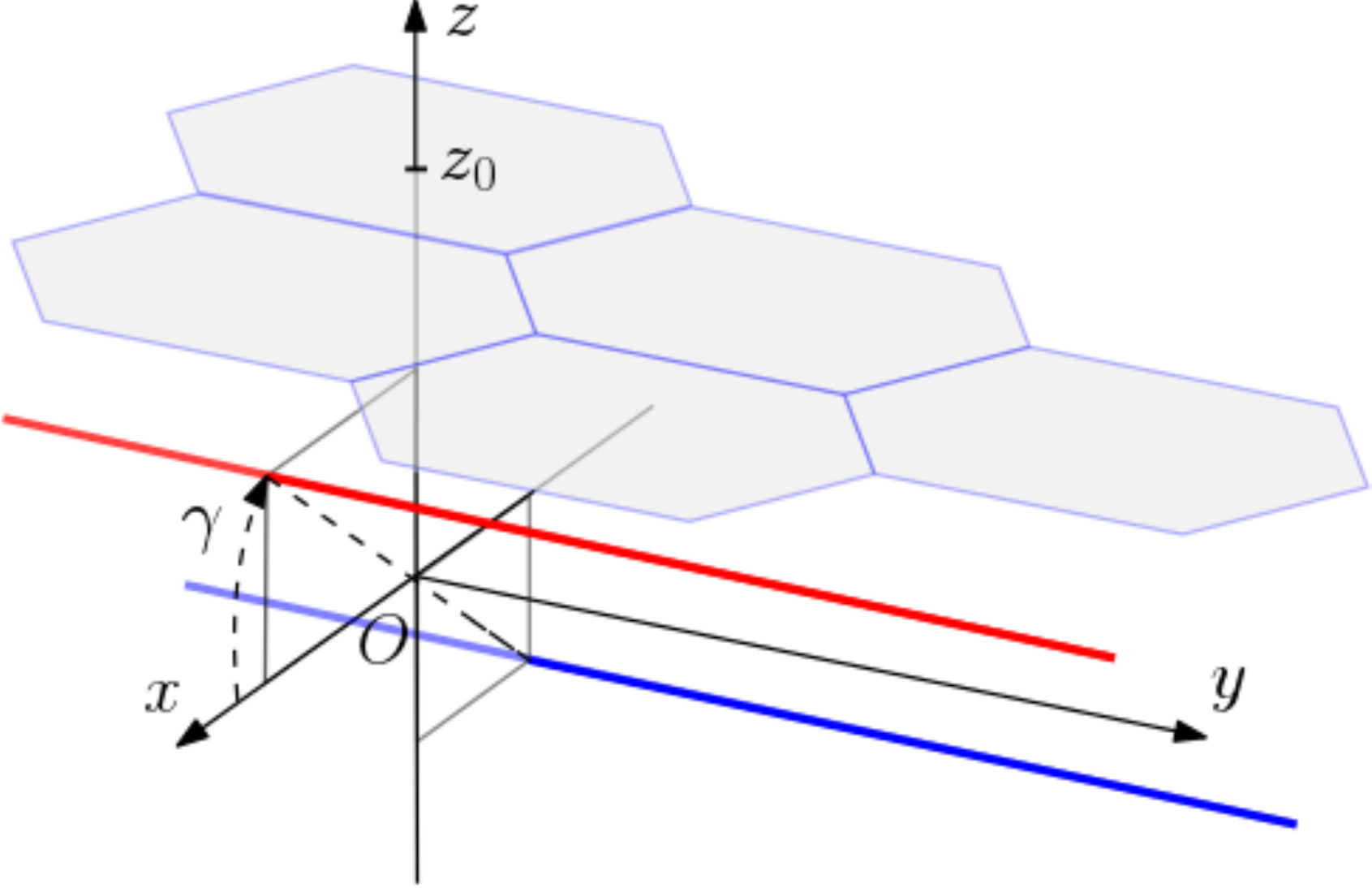} 
\caption{The thick red and blue lines represent a pair of wires which carry currents of the same magnitude but opposite orientations.} \label{fig:loop}
\end{figure}
%%%%%%%%%

One can calculate that
\begin{equation} \label{eq:rot_asy}
A_y(x)=\beta I_1\left(-\frac{4d\cos{\gamma}}{x}+\frac{4dz_0\sin{\gamma}}{x^2}+\mathcal{O}\left(\frac{1}{x^3}\right)\right)
\end{equation}
as $x\to\pm\infty$. If $\gamma\notin\{\frac{\pi}{2},\frac{3\pi}{2}\}$, i.e., 
the plane of the loop is not perpendicular to the sheet, the vector potential 
is not integrable. Using Proposition \ref{prop:slowly} we infer that there are infinitely many discrete eigenvalues of $H_{\mathbf{A}}[k_y]$ for any $k_y\neq 0$.

This changes dramatically if $\gamma\in\{\frac{\pi}{2},\frac{3\pi}{2}\}$. In that case, $A_y$ is of constant sign and the leading term in \eqref{eq:rot_asy} is proportional to $x^{-2}$. We can prove that there are discrete eigenvalues for all $k_y$ from some half-line and there are no discrete eigenvalues for all $k_y$ from another half-line. To be more specific, let us start by putting 
$\gamma=\frac{\pi}{2}$, i.e., with
\begin{equation*}
A_y(x)=\beta I_1\ln\frac{x^2+(z_0+d)^2}{x^2+(z_0-d)^2}.
\end{equation*}
Suppose $I_1>0$ and $z_0>d$. The first assumption amounts only to choosing a definite direction of the current, the latter means that the whole loop lies below the sheet.

Since $A_y$ is everywhere positive, $H_{\mathbf{A}}[k_y]$ has no eigenvalues for all $k_y\geq 0$, due to Proposition \ref{prop:no_ev}. As $A_y$ is also integrable, we can use the first criterion of Proposition \ref{prop:fastly} to prove the existence of discrete eigenvalues.
After a rather tedious calculation we arrive at
\begin{align*}
&\int_\R A_y(x)\dd x=4\pi\beta I_1 d,\\
&\int_\R A^2_y(x)\dd x=8\pi\beta^2 I_1^2 \left(z_0\ln\frac{z_0^2-d^2}{z_0^2}+d\ln\frac{z_0+d}{z_0-d}\right).
\end{align*}
Therefore, there are discrete eigenvalues in the spectrum of $H_{\mathbf{A}}[k_y]$, whenever
\begin{equation}\label{condperp1}
k_y<-\beta I_1\left(\ln\frac{z_0+d}{z_0-d}+\frac{z_0}{d}\ln\frac{z_0^2-d^2}{z_0^2}\right).
\end{equation}
Remark that the first term is the same as that we would obtain if we used Proposition \ref{prop:const_sign}, the second term further enlarges the range of admissible $k_y$'s.

The other case, $\gamma=\frac{3\pi}{2}$, is equivalent to the previous one, $\gamma=\frac{\pi}{2}$, after changing the orientation of the currents.  Therefore, for all 
 $$k_y>\beta I_1\left(\ln\frac{z_0+d}{z_0-d}+\frac{z_0}{d}\ln\frac{z_0^2-d^2}{z_0^2}\right),$$
there are some and, for all
 $$k_y\leq0,$$
there are none eigenvalues in the spectrum of $H_{\mathbf{A}}[k_y]$, respectively.

We found that when the plane of the loop and the sheet of the Dirac material are perfectly perpendicular, then  the waveguide might host unidirectional dispersionless wave packets. However, when the matching of the two planes is not perfect, the system may host bidirectional dispersionless wave packets. We will analyze in more detail  a special case of the loop that is slightly rotated away from its perpendicular position, i.e., we take $\gamma=\frac{\pi}{2}+\delta$, where $\delta$ is very small. Let us denote the potential term of $h_\pm[k_y]$ with constants subtracted by $V_\pm$. We have
$$V_\pm(x)=A_y(x)(A_y(x)+2k_y)\pm A_y'(x),$$
and so
\begin{equation} \label{eq:loop_asy}
V_\pm(x)=8d\beta I_1  k_y\left(\frac{\sin{\delta}}{x}+\frac{z_0\cos{\delta}}{x^2}\right)+\frac{(4d\beta I_1 \sin{\delta})^2\mp 4d\beta I_1 \sin{\delta}}{x^2}+\mathcal{O}\left(\frac{1}{x^3}\right)
\end{equation}
as $|x|\to\infty$.

Now, for $\delta=0$, if $k_y$  is sufficiently large then $V_\pm$ is everywhere positive. In fact, $V_\pm$ remains positive on a chosen interval $(-K,K)$ for all sufficiently small $\delta$'s. Let us fix $K$ so that \eqref{eq:loop_asy} is a good approximation outside $(-K,K)$. Then up to the second order in $x$ and the first order in $\delta$, the inequality $V_-<0$ is satisfied for $\delta>0$, provided that
$$\frac{\delta}{x}+\frac{z_0}{x^2}<0,$$
i.e. $x<-z_0/\delta$. For $\delta<0$, $V_+>0$ is satisfied, provided that $x>z_0/|\delta|$. Therefore, for small perturbations of the perfectly perpendicular configuration, the potential well emerges from infinity. The associated bound states are presumably localized near the potential well, i.e., far from the position of the wires.

\subsection{Tunable waveguide} \label{sec:tunable}

In this section, we propose a configuration of five wires such that their positions are fixed whereas by tuning the currents that they carry we can achieve different types of asymptotic behavior for the associated vector potential
\begin{equation*}
A_y(x)=\beta \sum_{j=1}^5I_j\ln((x-x_j)^2+z_j^2).
\end{equation*}
Some of the currents may be zero. In that case, we effectively deal with a lesser number of wires. Let us denote the $j$th wire by $W_j$ and fix their positions as follows,
\begin{equation*}
 x_1=x_4=-x_2=-x_5,\quad x_3=0,\quad z_1=z_2,\quad z_4=z_5,
\end{equation*}
where $0>z_1>z_4$. The wires $W_1,\,W_2,\,W_4$, and $W_5$ form an infinite cuboid below the plane $z=0$ of the sheet of the Dirac material, and $W_3$ is placed in the plane of vertical symmetry of this cuboid. See Table \ref{tab:wires} for figures of some typical cross-sections ($y=const.$).

Firstly, we choose  the currents in the following manner,
\begin{equation}\label{currents1}
I_1=I_2=-I_4=-I_5,\quad I_3=0.
\end{equation}
This corresponds to a pair of loops with common direction of currents that are both perpendicular to the sheet. We have
\begin{equation}\label{A1}
A_y(x)=\beta I_1 \ln\left(1+\frac{(z_1^2-z_4^2)(2x^2+2x_1^2+z_1^2+z_4^2)}{((x-x_1)^2+z_4^2)((x+x_1)^2+z_4^2)}\right)=\frac{\alpha_2}{x^2}+O\left(\frac{1}{x^4}\right)
\end{equation}
as $x\to\pm\infty$.
Due to the additivity of vector potential, one can expect similar result as in the case of a single loop that was treated within Section \ref{dipole}. In particular, the system might host unidirectional  dispersionless wave packets. Their direction may be switched  by changing the orientation of the currents.

If we switch the current orientation in one of the loops, i.e., we set 
\begin{equation}\label{currents2}
I_1=I_5=-I_2=-I_4,\quad I_3=0,
\end{equation}
then the situation changes dramatically, because the leading terms in the asymptotic expansions of vector potentials generated by single loops cancel out. Indeed, after some manipulations we obtain
\begin{equation}\label{A2}
A_y(x)=\beta I_1\ln\left(1-\frac{4xx_1(z_4^2-z_1^2)}{((x-x_1)^2+z_1^2)((x+x_1)^2+z_4^2)}\right)=\frac{\alpha_3}{x^3}+O\left(\frac{1}{x^5}\right)
\end{equation}
as $x\to\pm\infty$.
Therefore, by Proposition \ref{prop:gen}, the system possesses   bidirectional dispersionless wave packets, as there are discrete eigenvalues in the spectrum of $H_{\mathbf{A}}[k_y]$ for any $k_y$ with sufficiently large $|k_y|$. In view of Proposition \ref{prop:teschl}, there is finitely many of them.

Finally, we will present another simple configuration that makes the vector potential to decay even faster. We switch the current off in $W_1$ and $W_2$, and switch it on in $W_3$,  i.e., we set
\begin{equation}\label{currents3}
I_1=I_2=0,\quad I_4=I_5=-\frac{1}{2}I_3.
\end{equation}
Additionally, we require that
\begin{equation}\label{z1}
z_4=-\sqrt{x_4^2+z_3^2}. 
\end{equation}
As we suppose that the relative positions of the wires are fixed, the condition (\ref{z1}) restricts the distance of the system of wires from the sheet. The vector potential then acquires the following form,
\begin{equation}\label{A3}
A_y(x)=-\beta I_4\ln\left(1-\frac{4x_4^2(x_4^2+z_3^2)}{((x-x_4)^2+x_4^2+z_3^2)((x+x_4)^2+x_4^2+z_3^2)}\right)=\frac{\alpha_4}{x^4}+O\left(\frac{1}{x^6}\right)
\end{equation}
as $x\to\pm\infty$.
Clearly, $A_y$ is of constant sign, which depends on the current orientation. In view of Propositions \ref{prop:gen} and \ref{prop:no_ev}, the system may host unidirectional dispersionless   wave packets. Taking Proposition \ref{prop:teschl} into the account, we infer that if  $H_{\mathbf{A}}[k_y]$ has some eigenvalues then there is only finitely many of them. 

We put the derived properties of the discussed system together with the illustration of the wire configuration into Table \ref{tab:wires}.

%\pagestyle{empty}
%\begin{landscape}
\begin{table}[htbp] 
	\begin{center}
		\begin{tabular}{l|l|l|l|l}
			 Currents& Sample configuration& $A_y$& Asymptotics&  Spectral characteristics\\
			\hline
			\parbox{3.5cm}{$I_1=I_2=-I_4=-I_5,$\\  $I_3=0$}&
			\parbox[c]{3cm}{
				\includegraphics[scale=0.35]{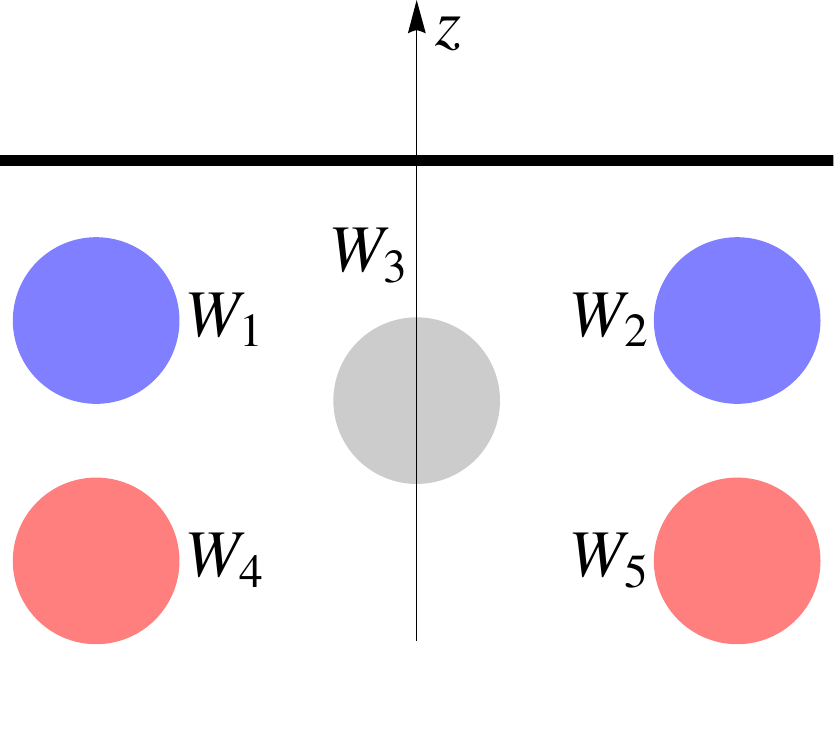}}&
			\eqref{A1}&
			\parbox{3.5cm}{$$A_y(x)=\frac{\alpha_2}{x^2}+O\left(\frac{1}{x^4}\right)$$}&
			\parbox[c]{4cm}{\rule{0pt}{2ex}\raggedright$\bullet$ unidirectional dispersionless wave packets\\ $\bullet$   finitness of the number of discrete eigenvalues of $H_{\mathbf{A}}[k_y]$ depends on particular values of $\alpha_2$ and $k_y$}\\
			\hline
			\parbox{3.5cm}{$I_1=I_5=-I_2=-I_4,$\\ $I_3=0$}&
			\parbox[c]{3cm}{
				\includegraphics[scale=0.35]{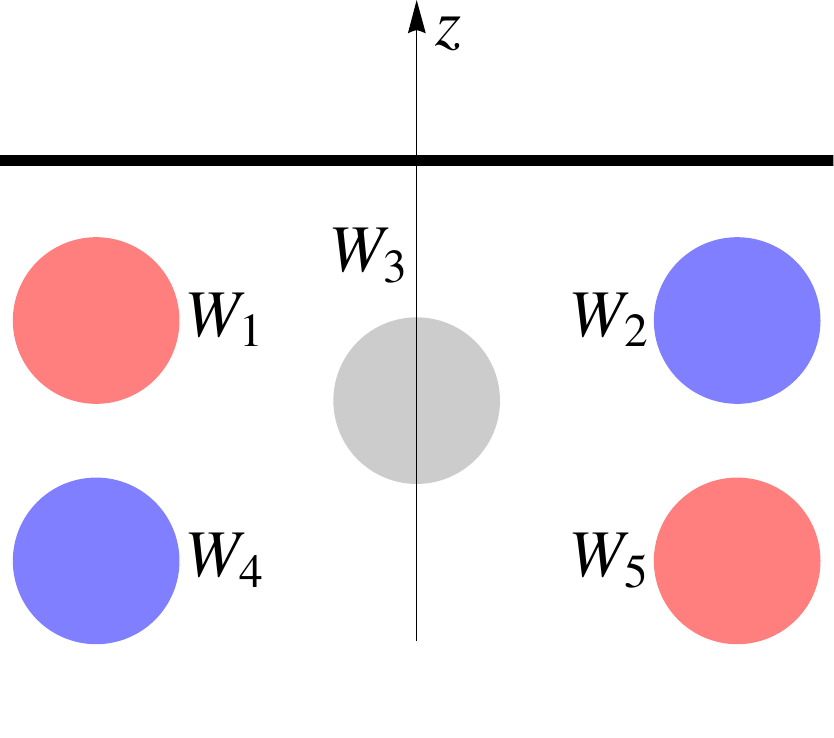}}&
			\eqref{A2}&
			\parbox{3.5cm}{$$A_y(x)=\frac{\alpha_3}{x^3}+O\left(\frac{1}{x^5}\right)$$}&
			\parbox[c]{4cm}{\raggedright $\bullet$ bidirectional dispersionless wave packets\\ $\bullet$ $H_{\mathbf{A}}[k_y]$ has at most finite number of discrete eigenvalues for any $k_y$}\\
			\hline
			\parbox{3.5cm}{$I_1=I_2=0,$\\ $I_4=I_5=-\frac{1}{2}I_3$}&
			\parbox[c]{3cm}{
				\includegraphics[scale=0.35]{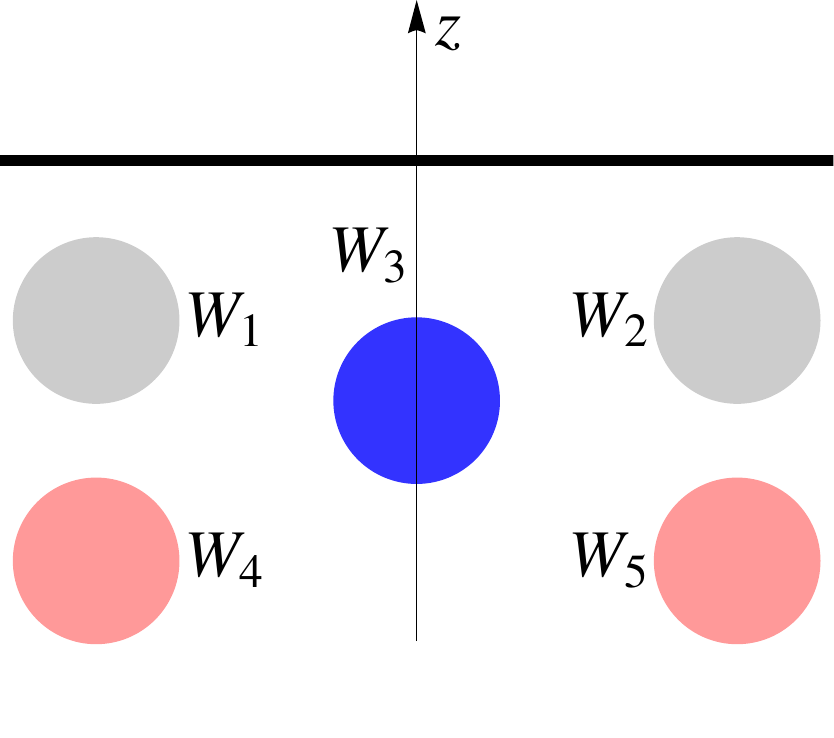}}&
			\eqref{A3}&
			\parbox{3.5cm}{$$A_y(x)=\frac{\alpha_4}{x^4}+O\left(\frac{1}{x^6}\right)$$}&
			\parbox[c]{4cm}{\raggedright $\bullet$ unidirectional  dispersionless wave packets\\ $\bullet$ $H_{\mathbf{A}}[k_y]$ has at most finite number of discrete eigenvalues for any $k_y$}\\
			\hline
		\end{tabular}
	\end{center}
	\caption{Regimes of the waveguide: Grey wires do not carry any current. Currents in the wires of the same color are identically oriented. Thick black line represents the sheet of the Dirac material. } \label{tab:wires}
\end{table} 
%\end{landscape}
\pagestyle{plain}

\section{Magnetized strips} \label{sec:strips}
We will study here spectral properties of $H_{\mathbf{A}}[k_y]$ with the magnetic field generated by infinitesimally thin magnetized strips of constant width $2d$ with the vector of magnetization directed either parallelly or perpendicularly to the strip.
Since the chosen gauge \eqref{eq:gauge} is always transverse, i.e., $\mathrm{div} \mathbf{A}=0$, the vector potential associated with a generic magnetized material with the magnetization $\mathbf{M}(\mathbf{r})$ can be written as 
\begin{equation} \label{eq:vp_magnetization}
\mathbf{A}(\mathbf{r})=\frac{\mu_0\sqrt{\alpha}}{4\pi}\int_{\mathbb{R}^3}\frac{\mathbf{M}(\mathbf{r}_0)\times (\mathbf{r}-\mathbf{r}_0)}{|\mathbf{r}-\mathbf{r}_0|^3}\dd\mathbf{r}_0.
\end{equation}

\subsection{Perpendicularly magnetized strip}

Firstly, consider the situation when the perpendicularly magnetized strip is parallel with the sheet of the Dirac material. For definiteness, we take $[-d,d]\times\R\times\{ 0\}$ for the strip and $\R^2\times\{z_0\}$ with $z_0>0$ for the sheet, see Figure \ref{fig:setting}. 
%%%%%%%%%
\begin{figure} 
\centering
\includegraphics[width=7cm]{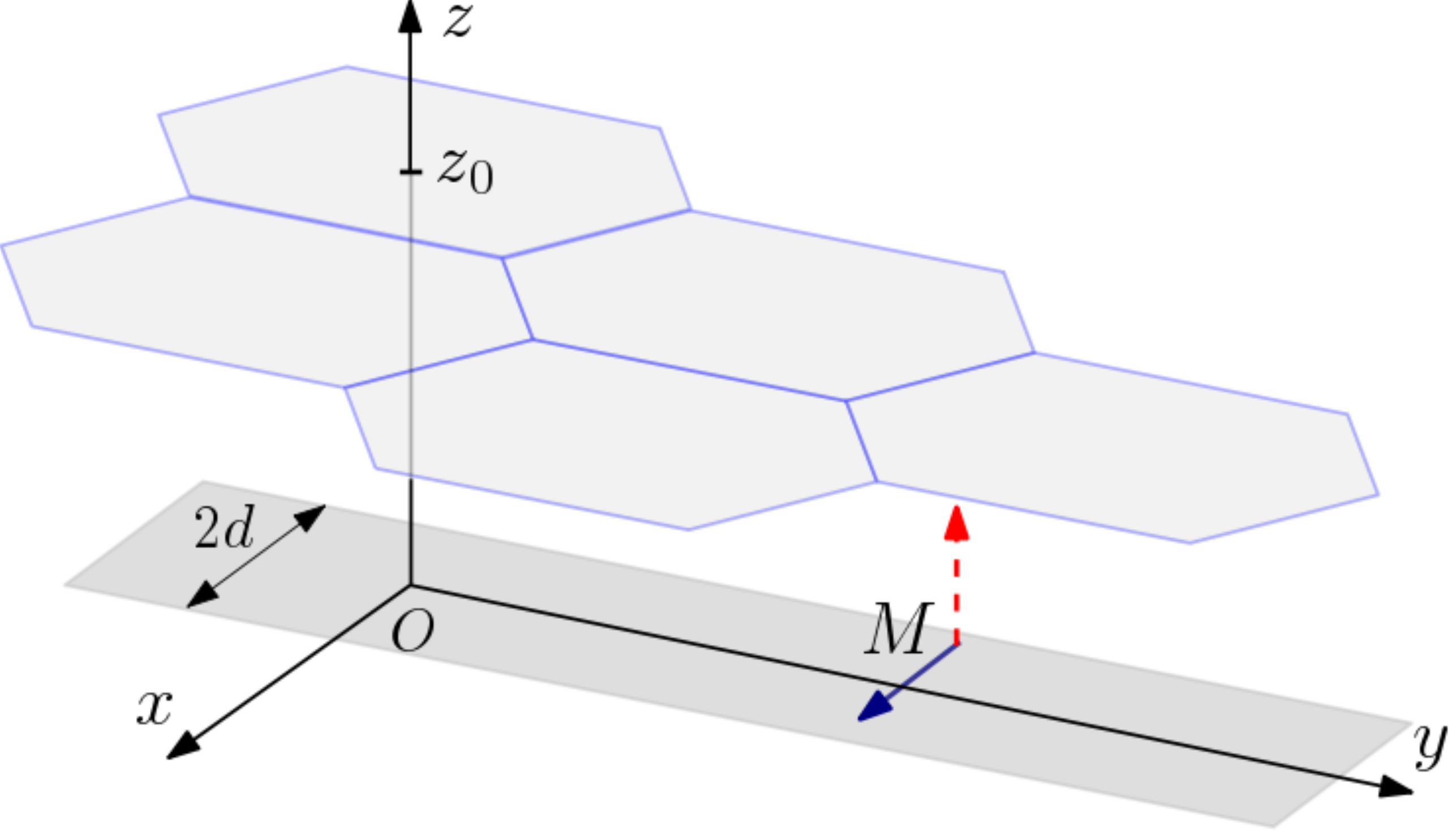} 
\caption{The blue full arrow depicts the parallel magnetization and the red dashed arrow corresponds to the perpendicular magnetization.} \label{fig:setting}
\end{figure}
%%%%%%%%%
The magnetization vector $\mathbf{M}$ is then given by
\begin{equation} \label{eq:per_magnetization}
\mathbf{M}(\mathbf{r})=(0,0,m \chi_{[-d,d]}(x)\delta(z)),
\end{equation}
where $m>0$.
Here, the indicator function $\chi_{[-d,d]}$ is defined as $\chi_{[-d,d]}(x)=1$ for all $x\in[-d,d]$ and it is zero elsewhere. 
Substituting \eqref{eq:per_magnetization} into \eqref{eq:vp_magnetization}, we get $A_x=A_z=0$, and, for all $z\neq 0$,
\begin{equation*}
A_y(x,z)=\frac{\tilde\beta}{2}\int_{-d}^d\dd x_0\int_{\R}\dd y_0\frac{x-x_0}{((x-x_0)^2+y_0^2+z^2)^{3/2}}=\frac{\tilde\beta}{2}\ln\left(\frac{z^2+(x+d)^2}{z^2+(x-d)^2}\right),
\end{equation*}
where we defined
$$\tilde\beta:=\frac{\mu_0 m\sqrt{\alpha}}{2\pi}.$$

We studied the same vector potential (with the current $-I_1$  substituted for the magnetization $m$) in the case of the current loop that consists of two infinite wires $\{\pm d\}\times\R\times\{0\}$, cf. \eqref{eq:A_rotated} with $\gamma=0$ and $z_0=z$. Rotating the strip along the $y$-axis by angel $\gamma$ is equivalent to rotating the wires to the position given by \eqref{eq:loop_pos}. Hence, we can use the results of Section \ref{dipole} literally. In particular, one can expect the unidirectional dispersionless wave packets only when the strip is exactly perpendicular to the sheet.

\subsection{Parallelly magnetized strip}
Firstly, let us assume again that the strip is parallel with the sheet and choose the same coordinates as in the previous section. If we take
\begin{equation}\label{parallelly}
\mathbf{M}(\mathbf{r})=(m \chi_{[-d,d]}(x)\delta(z),0,0) \end{equation}
with $m>0$ for the magnetization then we get $A_x=A_z=0$ and, for all $z\neq 0$,
\begin{equation}\label{newparalel}
A_y(x,z)=-\frac{\tilde\beta}{2}\int_{-d}^d\dd x_0\int_{\R}\dd y_0\frac{ z}{((x-x_0)^2+y_0^2+z^2)^{3/2}}=\tilde\beta\left(\arctan\frac{x-d}{z}-\arctan\frac{x+d}{z}\right).
\end{equation}

Next, let us rotate the strip around the $y$-axis by angle $\gamma\in[0,2\pi)$ measured from the positive $x$-axis towards the positive $z$-axis. This amounts to a natural change of coordinates,
\begin{equation*}
 \mathbf{r}'=R\mathbf{r}, \quad\text{with }R=\begin{pmatrix}
                                               \cos{\gamma} & 0 & -\sin{\gamma}\\
                                               0 & 1 & 0\\
                                               \sin{\gamma} & 0 & \cos{\gamma}
                                              \end{pmatrix}.
\end{equation*}
The vector potential $\mathbf{A}'$ associated with the rotated strip is given by \eqref{newparalel}, where we substitute $(x',z')$ for $(x,z)$. Therefore, going back to the original variables, we obtain
\begin{equation*}
 \mathbf{A}(\mathbf{r})=R^{-1}\mathbf{A}'(R\mathbf{r})=\tilde\beta\left(0,\arctan\frac{x\cos\gamma-z\sin\gamma-d}{x\sin\gamma+z\cos\gamma}-\arctan\frac{x\cos\gamma-z\sin\gamma+d}{x\sin\gamma+z\cos\gamma},0\right).
\end{equation*}
In particular, along the sheet placed at $z=z_0$ we get $A_x=A_z=0$ and
\begin{equation*}
 A_y(x)=\tilde\beta\left(\arctan\frac{x\cos\gamma-z_0\sin\gamma-d}{x\sin\gamma+z_0\cos\gamma}-\arctan\frac{x\cos\gamma-z_0\sin\gamma+d}{x\sin\gamma+z_0\cos\gamma}\right).
\end{equation*}
The magnetic field along the sheet is then given by
\begin{equation*}
 B(x)=\frac{1}{\sqrt{\alpha}}A_y'(x)=\frac{\mu_0 m}{2\pi}~\frac{2 d \sin\gamma \left(x^2+d^2-z_0^2\right)+4 d x z_0 \cos\gamma}{(x^2+z_0^2)^2+2 d^2 \big((z_0^2-x^2)\cos(2 \gamma)+2 x z_0 \sin(2\gamma)\big)+d^4}.
\end{equation*}

Looking at asymptotic expansions at $x=\pm\infty$,
\begin{equation} \label{eq:parall_asy}
 A_y(x)=\tilde\beta\left(-\frac{2d\sin\gamma}{x}-\frac{2dz_0\cos\gamma}{x^2}\right)+\mathcal{O}\left(\frac{1}{x^3}\right)
\end{equation}
and
\begin{equation} \label{eq:parall_B_asy}
 B(x)=\frac{\mu_0 m}{2\pi}~\frac{2d\sin\gamma}{x^2}+\mathcal{O}\left(\frac{1}{x^3}\right),
\end{equation}
we infer that the general assumptions of Theorem \ref{theo:gen} are fulfilled.

Formula \eqref{eq:parall_asy} suggests that $A_y$ is not integrable except for the angles $\gamma=0$ or $\gamma=\pi$, which correspond to the case when the strip is parallel with the sheet. Except this specific case, $H_{A}[k_y]$ has infinite number of discrete eigenvalues for any $k_y\neq 0$, which follows from Proposition \ref{prop:slowly}.

Let us now focus on the case when the strip is parallel with the sheet. For definiteness, put $\gamma=0$.  Using \eqref{eq:parall_asy} and \eqref{eq:parall_B_asy} with $\gamma=0$, we obtain
\begin{equation*}
 V_\pm(x)=A_y(x)(A_y(x)+2k_y)\pm A_y'(x)=-\frac{4\tilde\beta d z_0 k_y}{x^2}+\mathcal{O}\left(\frac{1}{x^3}\right)
\end{equation*}
as $x\to\pm\infty$. Therefore, due to Proposition \ref{prop:teschl}, 
$H_{\mathbf{A}}[k_y]$ has infinitely many discrete eigenvalues, whenever
\begin{equation} \label{eq:parall_cond_1}
 k_y>\frac{1}{16\tilde\beta d z_0},
\end{equation}
and at most finite number of them (possibly none), provided that 
\begin{equation} \label{eq:parall_cond_1b}
 k_y<\frac{1}{16\tilde\beta d z_0}.
\end{equation}

In the same moment, we can apply Proposition \ref{prop:const_sign}, because, for $\gamma=0$, $A_y$ is everywhere negative.  It says that,
for all $k_y$ satisfying
\begin{equation} \label{eq:parall_cond_2}
 k_y>\frac{1}{2}|\min A_y|=-\frac{1}{2}A_y(0)=\tilde\beta\arctan\frac{d}{z_0},
\end{equation}
there are discrete eigenvalues in the spectrum of $H_{\mathbf{A}}[k_y]$.
In fact, we can also use Proposition \ref{prop:fastly}, to provide a sharper, though more difficult to evaluate, bound on $k_y$. After some tedious calculation, one arrives at
\begin{align*}
 &\int_\R A_y(x)\dd x=-2\pi\tilde\beta d,\\
 &\int_\R A_y^2(x)\dd x=2\tilde\beta^2 z_0\int_\R \arctan{x}\left(\arctan{x}-\arctan\left(x+\frac{2d}{z_0}\right)\right)\dd x.
\end{align*}
Hence, another sufficient condition for the existence of discrete eigenvalues reads
\begin{equation} \label{eq:parall_cond_3}
 k_y>\frac{\tilde\beta z_0}{2\pi d}\int_\R \arctan{x}\left(\arctan{x}-\arctan\left(x+\frac{2d}{z_0}\right)\right)\dd x.
\end{equation}

We conclude that, for any $k_y$ that obeys \eqref{eq:parall_cond_1} or \eqref{eq:parall_cond_3} (or its more explicit, though weaker, version \eqref{eq:parall_cond_2}), there are discrete eigenvalues in the spectrum of $H_{\mathbf{A}}[k_y]$. There are infinitely many of them if \eqref{eq:parall_cond_1} is fulfilled, and finitely many if \eqref{eq:parall_cond_1b} holds together with \eqref{eq:parall_cond_3}. On the other hand, for $k_y\leq 0$, $H_{\mathbf{A}}[k_y]$ has no eigenvalues, because $A_y$ is everywhere negative and, consequently, one can apply Proposition \ref{prop:no_ev}.

The other case $\gamma=\pi$ amounts to changing the orientation of the magnetization, i.e. changing $\tilde\beta$ for $-\tilde\beta<0$. Consequently, there will always be discrete eigenvalues for all $k_y$ sufficiently negative, and there will not be any eigenvalues for all $k_y\geq 0$ . One can find  a threshold value of $k_y$ for the presence of eigenvalues in exactly the same manner as in the case $\gamma=0$.

Finally, let us remark that when the strip is tilted very slightly from its parallel position, i.e., $\gamma\neq0$ but $|\gamma|<<1$, one can perform an analysis analogous to that from the end of Section \ref{dipole} to conclude that tilting the strip produces a potential well that can accommodate  bound states but is localized far from the strip.

\section{Summary and Outlook}
Magnetic field can be used to create waveguides for two-dimensional Dirac fermions. 
The aim of the current paper was to get insight into spectral properties of both massive and massless Dirac fermions in presence of realistic magnetic fields. As the explicit solution of the associated stationary equation is not possible in these cases, we resolved to the qualitative spectral analysis. We arrived at general results applicable for a broad class of magnetic fields and employed them in the settings where the waveguides are created by placing the Dirac material into the proximity of a set of parallel wires with zero net current or of a magnetized strip.  

The analysis was facilitated by the translational symmetry of the system. The spectral properties of the two-dimensional Dirac Hamiltonian (\ref{H_A}) were fully deduced from those of the class of one-dimensional Hamiltonians $H_\mathbf{A}[k_y]$, see (\ref{Hky}), which is obtained by direct integral decomposition of the total Hamiltonian (\ref{H_A}) and describes dynamics of the system with the fixed longitudinal momentum $k_y$. We focused on a wide class of magnetic fields where the vector potential can be fixed so that it vanishes asymptotically together with its first derivative, see  \eqref{eq:B_vanishes} and \eqref{eq:A_vanishes}. Discrete energies  of $H_{\mathbf{A}}[k_y]$ were of our main interest;  they give rise to energy bands of the total Hamiltonian. These energy bands can be associated with dispersionless wave packets in the systems with translational symmetry  \cite{dispersionless}. 

In Section \ref{sec:general}, we discussed general spectral properties of $H_{\mathbf{A}}[k_y]$. In particular,  we presented a set of practical criteria that make it possible to guarantee existence or absence of discrete energies. They can also tell whether there is finite or infinite number of discrete energies in the spectrum of $H_{\mathbf{A}}[k_y]$ for a given $k_y$. We employed these qualitative tools extensively in the analysis of the magnetic waveguides generated by a set of wires with zero net electric current or magnetized strips.

In Section \ref{sec:wires}, we showed that magnetic field generated by a set of parallel wires with zero net current, placed in the proximity of  the Dirac material (e.g. graphene), can create magnetic  waveguide with tunable transport properties. Depending upon the currents in the wires and their positions, the waveguide can host either unidirectional or bidirectional dispersionless wave packets. We proposed a simple system with five wires whose transport properties are easily controllable by switching the currents in particular wires, cf. Table \ref{tab:wires}. 
In Section \ref{sec:strips}, we analyzed spectral properties of Dirac fermions in the magnetic waveguide produced by the magnetized strip. We considered both  the perpendicular (\ref{eq:per_magnetization}) and the parallel (\ref{parallelly}) magnetization of the strip. In both cases, the waveguide can again host either unidirectional or bidirectional dispersionless wave packets, depending upon the tilt of the strip.

Note that in our current model, no energy can enter the interval $(-M,M)$. It is worth mentioning in this context that there exist models of Dirac fermions with constant mass $M$ on the  half-plane where a specific energy band can cross the interval $(-M,M)$, see \cite{GrLe_06}. Existence of these bands is related to the boundary condition along the half-plane.

Our qualitative approach is complementary to the existing literature where the models offering explicit analytical  solutions are mostly discussed. 
It does not rely on the concrete form of the potential and it can provide valuable information on the spectrum in the cases where the explicit solutions are unknown. 
The presented results are applicable to two-dimensional massive or massless Dirac fermions in  presence of a magnetic field with translational symmetry. Their extension  to systems with other symmetries or with electric fields  would be desirable and should be addressed in the future.

\section*{Acknowledgments}
MF and MT were supported by GA\v CR grant No. 18-08835S and No. 17-01706S, respectively. VJ was supported by No.\ 15-07674Y.

\end{document}